\documentclass[12pt,a4paper]{article}

\usepackage{fullpage}
\usepackage{epsf}
\usepackage{amsmath,amssymb}
\usepackage{pstricks,pst-node,pst-text,pst-3d}

\newtheorem{theorem}{Theorem}
\newtheorem{lemma}{Lemma}
\newtheorem{proposition}{Proposition}

\newtheorem{definition}{Definition}
\newcounter{remark}
\newenvironment{remark}
{\begin{quote}\textsc{Remark} \stepcounter{remark} \arabic{remark}:}
{\end{quote}}
\newcounter{example}

\newenvironment{proof}{\medskip\noindent \bf Proof: \rm}{\hspace*{\fill}
$\blacksquare$ \newline \medskip}

\newcommand{\Q}{\mathcal{Q}}

\begin{document}

\title{Bipolarization of posets and natural interpolation\thanks{A short and preliminary
    version of this paper has been presented at the RIMS Symposium  on
    Information and mathematics of nonadditivity and nonextensivity, Kyoto
    University, Sept. 2006 \cite{grla06}.}}

\author{
  Michel GRABISCH\thanks{Corresponding author. Tel (+33) 1-44-07-82-85, Fax
    (+33) 1-44-07-83-01, 
    email \texttt{michel.grabisch@univ-paris1.fr}}\\
    Universit\'e de Paris I  -- Panth\'eon-Sorbonne \\
    Centre d'Economie de la Sorbonne\\ 106-112 Bd. de l'H\^opital, 75013 Paris,
  France \and 
    Christophe LABREUCHE\\
    Thales Research \& Technology \\
    RD 128, 91767 Palaiseau Cedex, France }

\date{Version of \today}

\maketitle

\begin{abstract}
The Choquet integral w.r.t. a capacity can be seen in the finite case as a
parsimonious linear interpolator between vertices of $[0,1]^n$. We take this
basic fact as a starting point to define the Choquet integral in a very general
way, using the geometric realization of lattices and their natural
triangulation, as in the work of Koshevoy. 

A second aim of the paper is to define a general mechanism for the
bipolarization of ordered structures. Bisets (or signed sets), as well as
bisubmodular functions, bicapacities, bicooperative games, as well as the
Choquet integral defined for them can be seen as
particular instances of this scheme. 

Lastly, an application to multicriteria aggregation with multiple reference
levels illustrates all the results presented in the paper. 

\vspace*{2cm} {\bf Keywords:} interpolation, Choquet integral, lattice, bipolar
structure
\end{abstract}

\newpage

\section{Introduction}\label{sec:intro}
Capacities and the Choquet integral \cite{cho53} have become fundamental
concepts in decision making (see, e.g., the works of Schmeidler \cite{sch86},
Murofushi and Sugeno \cite{musu91a}, and Koshevoy \cite{kos98}). 

An interesting but not so well known fact is that in the finite case, the
Choquet integral can be obtained as a parsimonious linear interpolation,
supposing that values on the vertices of the hypercube $[0,1]^n$ are known. The
interpolation formula was discovered by Lov\'asz \cite{lov83}, considering the
problem of extending the domain of pseudo-Boolean functions to
$\mathbb{R}^n$. Later, Marichal \cite{mar02} remarked that this formula was
precisely the Choquet integral (see also Grabisch \cite{gra03b}). 

The idea of considering the Choquet integral as a parsimonious linear interpolator
could serve as a basic principle for extending the notion of Choquet integral to
more general frameworks. An example of this has been done by the authors in
\cite{grla03b}, considering multiple reference levels in a context of
multicriteria aggregation. 

Another remarkable example of generalization of the Choquet integral is the one
for bicapacities, proposed by the authors \cite{grla03d,grla04}. As this paper
will make it clear, bicapacities are an example of concept based on the of
\emph{bipolarization} of a partially ordered set, in this case Boolean
lattices. Specifically, take a finite set $N$ and the set $2^N$ of all its
subsets ordered by inclusion: we obtain a Boolean lattice, and a capacity is an
isotone real-valued mapping on $2^N$. Introducing $\mathcal{Q}(N):=\{(A,B)\in
2^N\times 2^N\mid A\cap B=\varnothing\}$, a bicapacity is a real-valued mapping
on $\mathcal{Q}(N)$, satisfying some monotonicity condition. Observe that
$\mathcal{Q}(N)$ could be denoted by $3^N$ as well: $(A,B)\in\mathcal{Q}(N)$ can
be considered as a function $\xi$ of $\{-1,0,1\}^N$, where $\xi(i)=1$ if $i\in
A$, $\xi(i)=-1$ if $i\in B$, and 0 otherwise. Then the term ``bipolarization''
becomes clear, since $2^N\equiv \{0,1\}^N$ has one ``pole'' (the value 1, and 0
is the origin or neutral value), and $\{-1,0,1\}^N$ has 2 poles, namely $-1$ and
1, around the neutral value 0.

The set $3^N$ and functions defined on it are not new in the literature. To the
knowledge of the authors, it has been introduced approximately at the same time
and independently by Chandrasekaran and Kabadi \cite{chka88}, Bouchet
\cite{bou87}, Qi \cite{qi88}, and Nakamura \cite{nak88} in the field of matroid
theory and optimization, and later well developed by Ando and Fujishige
\cite{anfu96}. They use the term \emph{biset} or \emph{signed sets} for elements
of $3^N$, and \emph{bisubmodular functions} for bicapacities (with some more
restrictions).  In the field of cooperative game theory, Bilbao has introduced
\emph{bicooperative games} \cite{bil00}, which corresponds to bicapacities
without the monotonicity condition. Other remarkable works on bicapacities and
bicooperative games include the one of Fujimoto, who defined the M\"obius
transform of bicapacities under the name of \emph{bipolar M\"obius
transform}  \cite{fuj05}.

The aim of this paper is twofold: First to define the Choquet integral in a very
general way, as a parsimonious linear interpolation. This is done through the
concept of geometric realization of a distributive lattice and its natural
triangulation. Second, to provide a
general mechanism for the bipolarization of a poset, and to extend the previous
concepts (geometric realization, Choquet integral, etc.) to the bipolarized
structure. Then, all concepts around bisets, bicapacities, etc., are recovered as
a particular case.

Our work has been essentially inspired and motivated by Koshevoy, who used the
geometric realization of a lattice and its natural triangulation \cite{kos98},
and by Fujimoto \cite{fuj05}, who first remarked the inadequacy of our original
definition of the M\"obius transform for bicapacities in \cite{grla03d}, and
proposed the bipolar M\"obius transform. 

Section \ref{sec:prel} introduces necessary material, in particular geometric
realizations, natural triangulations and interpolation. Section \ref{sec:bipo}
is the core section of the paper, which presents the concept of bipolarization,
then the bipolar version of the geometric realization, natural triangulations
and interpolation. Lastly, Section \ref{sec:kary} gives some examples, and
develops the particular case of the product of linear lattices, which
corresponds to an application in multicriteria aggregation with reference
levels. We show that results obtained previously by the authors in
\cite{grla03b} are recovered.

\section{Preliminaries}\label{sec:prel}
In this section, we consider a finite index set $N:=\{1,\ldots,n\}$.

\subsection{Capacities and bicapacities}\label{sec:cabica}
We recall from Rota \cite{rot64} that, given a locally finite poset $(X,\leq)$
with bottom element, the \emph{M\"obius function} is the function $\mu:X\times
X\rightarrow \mathbb{R}$ which gives the solution to any equation of the form 
\begin{equation}
\label{eq:mob1}
g(x) = \sum_{y\leq x} f(y),
\end{equation}
for some real-valued functions $f,g$ on $X$, by
\begin{equation}
\label{eq:mob2}
f(x) = \sum_{y\leq x}\mu(y,x)g(y).
\end{equation}
Function $f$ is called the \emph{M\"obius transform} (or \emph{inverse}) of $g$.

\begin{definition}
  \begin{enumerate}
  \item A function $\nu:2^N\rightarrow\mathbb{R}$ is a \emph{game} if it
    satisfies $\nu(\varnothing)=0$.  
  \item A game which satisfies $\nu(A)\leq\nu(B)$  whenever $A\subseteq B$
    (monotonicity) is called a \emph{capacity} \cite{cho53} or \emph{fuzzy
      measure} \cite{sug74}. The
    capacity is \emph{normalized} if in addition $\nu(N)=1$.
  \end{enumerate}
\end{definition}
\emph{Unanimity games} are capacities of the type
\[
u_A(B):=\begin{cases}
  1, & \text{ if } B\supseteq A\\
  0, & \text{ else}
  \end{cases}
\]
for some $A\subseteq N, A\neq\varnothing$. It is well known that the set of
unanimity games is a basis for all games, whose coordinates in this basis are
exactly the M\"obius transform of the game.
\begin{definition}\label{def:cho}
Let us consider $f:N\rightarrow \mathbb{R}_+$. The
\emph{Choquet integral} of $f$ w.r.t.  a capacity $\nu$ is given by
\[
\int f\,d\nu:=\sum_{i=1}^n [f(\pi(i)) - f(\pi(i+1))]\nu(\{\pi(1),\ldots,\pi(i)\}),
\] 
where $\pi$ is a permutation on $N$ such that
$f(\pi(1))\geq\cdots\geq f(\pi(n))$, and $f(\pi(n+1)):=0$. 
\end{definition}
The above definition is valid if $\nu$ is a game. For any $\{0,1\}$-valued capacity
$\nu$ on $2^N$ we have (see, e.g., \cite{musu93}):
\begin{equation}\label{eq:musu}
\int f\,d\nu=\bigvee_{A\mid \nu(A)=1}\bigwedge_{i\in A}f(i).
\end{equation}
The expression of the Choquet
integral w.r.t. the M\"obius transform of $\nu$ (denoted by $m$) is
\begin{equation}\label{eq:chom}
\int f\, d\nu = \sum_{A\subseteq N}m(A)\bigwedge_{i\in A} f(i).
\end{equation}

We introduce $\Q(N):=\{(A,B)\in 2^N\times 2^N\mid A\cap B=\varnothing\}$. 
\begin{definition}
\begin{enumerate}
\item A mapping $v:\Q(N)\rightarrow\mathbb{R}$ such that
  $v(\varnothing,\varnothing)=0$ is called a \emph{bicooperative game} \cite{bil00}. 
\item A bicooperative game $v$ such that $v(A,B)\leq v(C,D)$ whenever
  $(A,B),(C,D)\in\Q(N)$ with $A\subseteq C$ and $B\supseteq D$ (monotonicity) is
  called a \emph{bicapacity}\cite{grla02a,grla03d}. Moreover, a bicapacity is
  normalized if in addition $v(N,\varnothing)=1$ and $v(\varnothing,N)=-1$.
\end{enumerate}
\end{definition}

\begin{definition}
\label{def:bicho}
Let $v$ be a bicapacity and $f$ be a real-valued function on $N$. The
\emph{(general) Choquet integral} of $f$ w.r.t $v$ is given by
\[
\int f\,dv := \int|f|\,d\nu_{N^+_f}
\] 
where $\nu_{N^+_f}$ is a game on $N$ defined by 
\[
\nu_{N^+_f}(C) := v(C\cap N^+_f, C\cap N^-_f), \quad \forall C\subseteq N
\]
and $N^+_f:=\{i\in N\mid f(i)\geq 0\}$, $N^-_f=N\setminus N^+_f$. 
\end{definition}
 Note that the definition remains valid if $v$ is a bicooperative game.

\medskip

Considering on $\Q(N)$ the product order 
\[
(A,A')\subseteq(B,B') \Leftrightarrow A\subseteq B \text{ and } A'\subseteq B',
\]
the M\"obius transform $b$ of a bicapacity $v$ is the solution of:
\begin{align*}
v(A_1,A_2) &
=\sum_{(B_1,B_2)\subseteq(A_1,A_2)}b(B_1,B_2)\\
& = \sum_{\substack{B_1\subseteq A_1\\ B_2\subseteq A_2}}b(B_1,B_2).
\end{align*}
This gives:
\[
b(A_1,A_2)=\sum_{\substack{B_1\subseteq A_1\\ B_2\subseteq A_2}}
(-1)^{|A_1\setminus B_1|+|A_2\setminus B_2|}v(B_1,B_2)
\]
(see Fujimoto \cite{fuj05}). Unanimity games are then naturally defined by
\[
u_{(A_1,A_2)}(B_1,B_2):=\begin{cases}
                        1, & \text{ if } (B_1,B_2)\supseteq(A_1,A_2)\\
                        0, & \text{ else.}
                        \end{cases}
\]
and form a basis of bicooperative games.

The expression of the Choquet integral in terms of $b$ is given by
\begin{equation}\label{eq:2}
\int f\, dv = \sum_{(A_1,A_2)\in\mathcal{Q}(N)} b(A_1,A_2)\Big[\bigwedge_{i\in
    A_1} f^+(i)\wedge \bigwedge_{j\in A_2}f^-(j)\Big],
\end{equation}
with $f^+:=f\vee 0$ and $f^-:=(-f)^+$. 

\subsection{Lattices, geometric realizations, and triangulation}\label{sec:latt}
A \emph{lattice} is a set $L$ endowed with a partial order $\leq$ such that for
any $x,y\in L$ their least upper bound $x\vee y$ and greatest lower bound
$x\wedge y$ always exist. For finite lattices, the greatest element of $L$
(denoted $\top$) and least element $\bot$ always exist. $x$ \emph{covers} $y$
(denoted $x\succ y$) if $x> y$ and there is no $z$ such that $x>z>y$. A sequence
of elements $x\leq y\leq\cdots\leq z$ of $L$ is called a \emph{chain} from $x$
to $z$, while an \emph{antichain} is a sequence of elements such that it
contains no pair of comparable elements. A chain from $x$ to $z$ is
\emph{maximal} if no element can be added in the chain, i.e., it has the form
$x\prec y\prec\cdots\prec z$.

The lattice is \emph{distributive} if $\vee,\wedge$ obey
distributivity. An element $j\in L$ is \emph{join-irreducible} if it cannot be
expressed as a supremum of other elements. Equivalently, $j$ is join-irreducible
if it covers only one element. Join-irreducible elements covering $\bot$ are
called \emph{atoms}, and the lattice is \emph{atomistic} if all
join-irreducible elements are atoms. The set of all join-irreducible elements
of $L$ is denoted $\mathcal{J}(L)$. 

For any $x\in L$, we say that $x$ \emph{has a complement in $L$} if there exists
$x'\in L$ such that $x\wedge x'=\bot$ and $x\vee x'=\top$. The complement is
unique if the lattice is distributive. 

An important property is that in a distributive lattice, any element $x$ can be
written as an irredundant supremum of join-irreducible elements in a unique
way. We denote by $\eta(x)$ the \emph{(normal)
decomposition} of $x$, defined as the set of join-irreducible elements smaller
or equal to $x$, i.e., $\eta(x):=\{j\in \mathcal{J}(L)\mid j\leq x\}$.  Hence
\[
x= \bigvee_{j\in \eta(x)} j
\]
(throughout the paper, $j,j',\ldots$ will always denote join-irreducible elements).
Note that this decomposition may be redundant. 

We can rephrase differently the above result in several ways, which will be useful
for the sequel.  $Q\subseteq L$ is a
\emph{downset} of $L$ if $x\in Q$, $y\in L$ and $y\leq x$ imply $y\in Q$. For
any subset $P$ of $L$, we denote by $\mathcal{O}(P)$ the set of all downsets of
$P$. Then the mapping $\eta$ is an isomorphism of $L$ onto
$\mathcal{O}(\mathcal{J}(L))$ (Birkhoff's theorem \cite{bir33}). Also,
\begin{equation}\label{eq:dist}
\eta(x\vee y)= \eta(x)\cup\eta(y), \quad \eta(x\wedge y)=\eta(x)\cap \eta(y)
\end{equation}
if $L$ is distributive. Next, downsets of some
partially ordered set $P$ correspond bijectively to nonincreasing mappings from
$P$ to $\{0,1\}$. Let us denote by $\mathcal{D}(P)$ the set of all nonincreasing
mappings from $P$ to $\{0,1\}$. Then Birkhoff's theorem can be rephrased as
follows: \emph{any distributive lattice $L$ is isomorphic to
  $\mathcal{D}(\mathcal{J}(L))$}. Finally, note that a mapping of
$\mathcal{D}(P)$ can be considered as a vertex of $[0,1]^{|P|}$. In summary, we have:
\begin{equation}
x\in L \leftrightarrow \eta(x)\in\mathcal{O}(\mathcal{J}(L)) \leftrightarrow
1_{\eta(x)}\in \mathcal{D}(\mathcal{J}(L)) \leftrightarrow (1_{\eta(x)},0_{\eta(x)^c})\in[0,1]^{|\mathcal{J}(L)|}
\end{equation}
where the notation $(1_A,0_{A^c})$ denotes a vector whose coordinates are 1 if
in $A$, and 0 otherwise.  All arrows represent isomorphisms, the leftmost one
being an isomorphism if $L$ is distributive.

\medskip

We introduce now the notion of \emph{geometric realization} of a lattice, and
its natural triangulation (see Koshevoy \cite{kos98} for more details). 
For any partially ordered set $P$, we define $\mathcal{C}(P)$ as the set of
nonincreasing mappings from $P$ to $[0,1]$. It is a convex polyhedron, whose set of vertices is
$\mathcal{D}(P)$.
\begin{definition}\label{def:geore}
The \emph{geometric realization} of a distributive
lattice $L$ is the set $\mathcal{C}(\mathcal{J}(L))$. 
\end{definition}
The \emph{natural triangulation} of
$\mathcal{C}(\mathcal{J}(L))$ consists in
partitioning $\mathcal{C}(\mathcal{J}(L))$ into simplices whose vertices are in
$\mathcal{D}(\mathcal{J}(L))$. These simplices correspond to maximal chains of
$\mathcal{D}(\mathcal{J}(L))$. The following proposition summarizes all what we
need in the sequel. 
\begin{proposition}\label{prop:rem2}
Suppose that $L$ is distributive, with $n$ join-irreducible elements. Consider
any maximal chain $C:=\{1_{\varnothing}=0\prec 1_{X_1}\prec\cdots\prec
1_{X_{|\mathcal{J}(L)|}}=1\}$. Then
\begin{enumerate}
\item The simplex $\sigma(C)$ is $n$-dimensional, and contains vertices
$(0,\ldots,0)$ and $(1,\ldots,1)$ in $[0,1]^n$.
\item The sequence
$X_1,\ldots,X_n$ induces a permutation $\pi:\{1,\ldots,n\}\rightarrow
  \mathcal{J}(L)$ such that $X_i=\{\pi(1),\ldots,\pi(i)\}$,
  $i=1,\ldots,n$, and
\begin{equation}\label{eq:fj}
f(j) = \sum_{i=1}^n \alpha_i 1_{X_i}(j) = \sum_{X_i\ni j}\alpha_i =
\sum_{i=\pi^{-1}(j)}^n \alpha_i,\quad \forall j\in\mathcal{J}(L).
\end{equation}
Conversely, a permutation $\pi$ induces a maximal chain if and only if it
  fulfills the condition
\[
\forall j,j'\in \mathcal{J}(L), j\leq j'\Rightarrow \pi^{-1}(j)\leq \pi^{-1}(j').
\]
\item The solution of (\ref{eq:fj}) is
\begin{equation}\label{eq:alpha}
\alpha_i = f(\pi(i)) - f(\pi(i+1)),\quad i=1,\ldots,n-1, \text{ and }
\alpha_n=f(\pi(n)),
\end{equation}
and $\alpha_0=1-\sum_{i=1}^n\alpha_i= 1 - f(\pi(1))$. In addition, $f(\pi(1))\geq
f(\pi(2))\geq\cdots\geq f(\pi(n))$.
\end{enumerate}
\end{proposition}

\begin{definition}
For any functional $F:\mathcal{D}(\mathcal{J}(L))\rightarrow\mathbb{R}$ on a
distributive lattice $L$, its \emph{natural extension to the geometric
realization of $L$} is defined by:
\[
\overline{F}(f):=\sum_{i=0}^p \alpha_i F(1_{X_i}) 
\]
for all $f\in \mathrm{int}(\sigma(C))$, with $C$ being a chain
$\{1_{X_0}<1_{X_1}<\cdots< 1_{X_p}\}$ in $\mathcal{D}(\mathcal{J}(L))$, and
$\sigma(C)$ its convex hull in $\mathcal{C}(\mathcal{J}(L))$, with
$f=\sum_{i=0}^p \alpha_i1_{X_i}$. 
\end{definition}
The following proposition readily follows from Proposition \ref{prop:rem2} and
the above definition.
\begin{proposition}\label{prop:rem3}
Let $L$ be a distributive lattice, with $n$ join-irreducible elements, and any
functional $F:\mathcal{D}(\mathcal{J}(L))\rightarrow\mathbb{R}$. Consider
any maximal chain $C:=\{1_{\varnothing}=0\prec 1_{X_1}\prec\cdots\prec
1_{X_{|\mathcal{J}(L)|}}=1\}$. 
\begin{enumerate}
\item For any $f\in\sigma(C)$,
\begin{equation}\label{eq:natin}
\overline{F}(f) =  \sum_{i=1}^n [f(\pi(i)) -
  f(\pi(i+1))]F(1_{\{\pi(1),\ldots,\pi(i)\}})
\end{equation}
with $f(\pi(n+1)):=0$.
\item $\overline{F}$ is linear in each simplex $\sigma(C)$, i.e., $\overline{F}(f+g)
  = \overline{F}(f) + \overline{F}(g)$ provided that $f,g,f+g$ belongs to the
  same $\sigma(C)$. Moreover, $\overline{F}$ is linear in $F$, in the sense that
  $\overline{F+G}(f)=\overline{F}(f) +\overline{G}(f)$ for any $f$. 
\end{enumerate}
\end{proposition}
\begin{quote}
\textbf{Example 1:} If $L$ is the Boolean lattice $2^N$, with $N:=\{1,\ldots,n\}$,
then $\mathcal{J}(L)=N$ (atoms). We have
$\mathcal{D}(\mathcal{J}(L))=\{x:N\rightarrow \{0,1\},x\text{
nonincreasing}\}$, but since $N$ is an antichain, there is no restriction on
$x$ and $\mathcal{D}(\mathcal{J}(L))=\{0,1\}^N$, i.e., it is the set of
vertices of $[0,1]^n$. Similarly, $\mathcal{C}(\mathcal{J}(L))= [0,1]^N$, which
is the hypercube itself. 

Consider now a maximal chain in
  $\mathcal{D}(\mathcal{J}(L))$, denoted by $C:=\{1_{A_0}<1_{A_1}<\cdots
  <1_{A_n}\}$, with $\emptyset=:A_0\subset A_1\subset\cdots\subset A_n:=N$.
  It corresponds to a permutation $\pi$ on $N$, with $A_i=\{\pi(1),\ldots,\pi(i)\}$.
Since $\mathcal{J}(L)$ is an antichain, conversely any permutation corresponds
  to a maximal chain. Using (\ref{eq:alpha}), we get
\begin{align*}
  \overline{F}(f) & = \sum_{i=1}^n\alpha_iF(1_{A_i})\\
  & = \sum_{i=1}^n [f(\pi(i)) -
  f(\pi(i+1))]F(1_{\{\pi(1),\ldots,\pi(i)\}}),
\end{align*}
with the convention $f(\pi(n+1)):=0$. Putting $\mu(A):=F(1_A)$, we recognize
the Choquet integral $\int f \,d\nu$ (see Definition \ref{def:cho}). $\Box$
\end{quote}
This example shows that the Choquet integral is the natural extension of
capacities. Hence, by analogy, $\overline{F}(f)$ could be
  called the \emph{Choquet integral of $f$ w.r.t. $F$}. Moreover, using Remark 1, we
  could consider $F$ as a game or capacity defined over a sublattice of the
  Boolean lattice $2^n$.

\section{Bipolar structures}\label{sec:bipo}
\subsection{Bipolar extension of $L$}\label{sec:biex}
\begin{definition}
Let us consider $(L,\leq)$ an inf-semilattice with bottom element $\bot$. The
\emph{bipolar extension} $\widetilde{L}$ of $L$ is defined as follows:
\[
\widetilde{L}:=\{(x,y)\mid x,y\in L, x\wedge y=\bot\},
\]
which we endow with the product order $\leq$ on $L^2$.   
\end{definition}
Remark that $\widetilde{L}$ is a downset of $L^2$. The following holds.
\begin{proposition}
Let $(L,\leq)$ be an inf-semilattice.
\begin{itemize}
\item [(i)] $(\widetilde{L},\leq)$ is an inf-semilattice whose bottom element
is $(\bot,\bot)$, where $\leq$ is the product order on $L^2$.
\item [(ii)] The set of join-irreducible elements of $\widetilde{L}$ is
\[
\!\!\!\!\mathcal{J}(\widetilde{L})=\{(j,\bot)\mid
j\in\mathcal{J}(L)\}\cup\{(\bot,j)\mid j\in\mathcal{J}(L)\}.
\]
\item [(iii)] The normal decomposition writes
\[
(x,y) = \bigvee_{j\leq x,j\in \mathcal{J}(L)}(j,\bot)\vee\bigvee_{j\leq y,j\in
\mathcal{J}(L)}(\bot,j). 
\]
\end{itemize}
\end{proposition} 
\begin{proof}
(i) Let us consider $(x,y),(z,t)\in L^2$. Then $(x,y)\wedge(z,t)=(x\wedge
z,y\wedge t)$ is the greatest lower bound of $(x,y)$ and $(z,t)$ for the
product order. Suppose $x\wedge y=\bot$ and $z\wedge t=\bot$. Then $(x\wedge
z)\wedge(y\wedge t)=\bot$ too, which proves that the greatest lower bound
always exists in $\widetilde{L}$. 

(ii) clear since these are the join-irreducible element of $L^2$, and they all
belong to $\widetilde{L}$. 

(iii) clear from (ii).
\end{proof}

We consider now the M\"obius function over $\widetilde{L}$. The aim is to solve
\begin{equation}
\label{eq:bm}
f(x,y) =\!\!\! \sum_{(x',y')\leq(x,y), (x',y')\in\widetilde{L}}g(x',y'),\quad\forall (x,y)\in\widetilde{L},
\end{equation}
where $f,g$ are real-valued functions on $\widetilde{L}$. The solution is given
through the M\"obius function on $\widetilde{L}$:
\begin{equation}
\label{eq:bms}
g(x,y) = \sum_{\substack{(z,t)\leq(x,y)\\(z,t)\in \widetilde{L}}}f(z,t)\mu_{\widetilde{L}}((z,t),(x,y)).
\end{equation}
The following holds.
\begin{proposition}\label{prop:mob}
The M\"obius function on $\widetilde{L}$ is given by:
\[
\mu_{\widetilde{L}}((z,t),(x,y))=\mu_L(z,x)\mu_L(t,y).
\]
\end{proposition}
\begin{proof}
Let us define $h(x',y):=\sum_{\substack{y'\leq y\\x'\wedge y'=\bot}}g(x',y')$
for a given $x'\in L$ such that $x'\wedge y=\bot$. Since $y'\leq y$, $x'\wedge
y=\bot$ implies $x'\wedge y'=\bot$ too. Hence:
\begin{equation}
\label{eq:e1}
h(x',y) = \sum_{y'\leq y}g(x',y').
\end{equation}
By a similar argument, note that (\ref{eq:bm}) can be rewritten as
\begin{equation}
\label{eq:bm'}
f(x,y) = \sum_{x'\leq x}\sum_{y'\leq y}g(x',y').
\end{equation}
Putting (\ref{eq:e1}) in (\ref{eq:bm'}) gives
\begin{equation}
\label{eq:e2}
f(x,y)= \sum_{x'\leq x}h(x',y). 
\end{equation}
Applying M\"obius inversion to (\ref{eq:e1}) and (\ref{eq:e2}) gives
\begin{equation}
\label{eq:e1'}
g(x,y) = \sum_{t\leq y}\mu_L(t,y)h(x,t), 
\end{equation}
for some fixed $x$, $x\wedge y=\bot$, and 
\begin{equation}
\label{eq:e2'}
h(x,y) = \sum_{z\leq x}\mu_L(z,x)f(z,y)
\end{equation}
for some fixed $y$, $x\wedge y=\bot$. Using (\ref{eq:e2'}) into (\ref{eq:e1'})
leads to, for $(x,y)\in\widetilde{L}$: 
\begin{align*}
g(x,y) &= \sum_{t\leq y}\mu_L(t,y)\sum_{z\leq x}\mu_L(z,x)f(z,y)\\
        &= \sum_{(z,t)\leq (x,y)}\mu_L(z,x)\mu_L(t,y)f(z,y).
\end{align*}
Note that in the last equation $(z,t)\in\widetilde{L}$ since $z\leq x$, $t\leq
y$ and $x\wedge y=\bot$ imply $z\wedge t=\bot$. Comparing the above last
equation with (\ref{eq:bms}) gives the desired result.
\end{proof}

Note that as usual, the set of functions $u_{(x,y)}$ defined
by
\begin{equation}\label{eq:ug}
u_{(x,y)}(z,t)=\begin{cases}
        1,\text{ if } (z,t)\geq(x,y)\\
        0,\text{ otherwise}     
        \end{cases}
\end{equation}
forms a basis of the functions on $\widetilde{L}$.

\begin{theorem}
Let $L$ be a finite distributive lattice, and $c(L)$ be the set of its
complemented elements. Then, for any $x\in c(L)$, its complement being denoted
by $x'$, the interval $L(x)$ of $\widetilde{L}$ defined by
\[
L(x):=[(\bot,\bot),(x,x')]
\] 
and endowed with the product order of $L^2$ is isomorphic to $L$, by the order
isomorphism $\phi_x:L(x)\rightarrow L$, $(y,z)\mapsto y\vee z$. The inverse
function $\phi_x^{-1}$ is given by $\phi_x^{-1}(w)=(w\wedge x,w\wedge x')$. 

Moreover, the join-irreducible elements of $L(x)$ are the image of those of $L$
by $\phi_x^{-1}$, i.e.:
\[
\mathcal{J}(L(x))= \{(j\wedge x,j\wedge x')\mid j\in \mathcal{J}(L)\}.
\]
\end{theorem}
\begin{proof}
Take $x\in c(L)$ and show that $\phi_x$ is an order isomorphism between $L(x)$
and $L$. First remark that if $y,z\in L$, then $y\vee z\in L$ since $L$ is a
lattice. Also for any $(y,z)\in L(x)$, since $y\leq x$, we have
$\eta(y)\subseteq \eta(x)$, and similarly $\eta(z)\subseteq \eta(x')$.

Let us show that $\phi_x$ is a bijection. Observe that since $x\wedge x'=\bot$
and $x\vee x'=\top$, we have $\eta(x)\cap \eta(x')=\varnothing$ and
$\eta(x)\cup\eta(x')=\mathcal{J}(L)$ by (\ref{eq:dist}), i.e., $x$ and $x'$
partition the join-irreducible elements of $L$. It follows that any $w\in L$ can
be written uniquely as $w=y\vee z$, with $y,z\in L$ defined by
\begin{equation}\label{eq:eta}
\eta(y)=\eta(w)\cap \eta(x),\quad \eta(z)=\eta(w)\cap \eta(x').
\end{equation}
Then $(y,z)\in L(x)$ since $\eta(y)\subseteq \eta(x)$ and $\eta(z)\subseteq
\eta(x')$. The expression of the inverse isomorphism
$\phi^{-1}_x(w)=(w\wedge x,w\wedge x')$ is clear from (\ref{eq:eta}) and
(\ref{eq:dist}).

Take $(y,z)\leq(y',z')$. This means $y\leq y'$ and $z\leq z'$, hence $y\vee
z\leq y'\vee z'$. Conversely, take $w\leq w'$. We have $y=w\wedge x\leq w'\wedge
x = y'$ and similarly for $z=w\wedge x'$. Hence $\phi_x$ is an order
isomorphism. 

Finally, since $\phi_x$ is an order isomorphism, the two
lattices $L$ and $L(x)$ have the same structure, and hence the same
join-irreducible elements. 
\end{proof}

Remark that in any finite lattice, $\bot$ and $\top$ are complemented elements,
and $L(\top) = L$, $L(\bot)=L^*$, where $L^*$ is the dual of $L$ (i.e., $L$ with
the reverse order). An
interesting question is whether the union of all $L(x)$, $x\in c(L)$, is equal
to $\widetilde{L}$.
\begin{theorem}\label{th:biex}
Let $L$ be a finite distributive lattice. Then the bipolar extension
$\widetilde{L}$ can be written as:
\[
\widetilde{L} = \bigcup_{x\in c(L)}L(x)
\]
if and only if $\mathcal{J}(L)$ has all its connected components with a single
bottom element.
\end{theorem}
\begin{proof}
  Take $(y,z)\in \widetilde{L}$, i.e., $y,z\in L$ and
  $\eta(y)\cap\eta(z)=\varnothing$. To find $x\in c(L)$  such that
  $(y,z)\in L(x)$ is equivalent to satisfy the conditions
\begin{enumerate}
\item $\mathcal{J}(L)\setminus\eta(x)$ is a downset   ($x$ is complemented)
\item $\eta(x)\supseteq \eta(y)$, and $\eta(x)\cap\eta(z)=\varnothing$ ($(y,z)$
  belongs to $L(x)$).
\end{enumerate}
Consider $\mathcal{J}(L)$. Its Hasse diagram is formed of connected components,
say $J_1,\ldots,J_l$. Remark that in a given connected component $J_k$, it is
not possible to partition it into downsets. Indeed, suppose that $J_k=D_1\cup
D_2$, with $D_1,D_2$ two disjoint nonempty downsets. Since $J_k$ is connected,
each $x\in J_k$ is comparable with another $y\in J_k$. Hence, by nonemptiness
assumption, there exists $x_1\in D_1$ which is comparable with some $x_2\in
D_2$, i.e., either $x_1\leq x_2$ or the converse. But then $x_1\in D_2$ (or
$x_2\in D_1$), which contradicts the fact they are disjoint. This proves that
complemented elements $x\in L$ are such that
\begin{equation}\label{eq:tile}
\eta(x)=\cup_{k\in K(x)}J_k
\end{equation}
for some index set $K(x)\subseteq \{1,\ldots,l\}$. 

Take some $(y,z)\in \widetilde{L}$ and suppose that $\eta(y)\subseteq \cup_{k\in
  K(y)}J_k$ and $\eta(z)\subseteq \cup_{k\in K(z)}J_k$. Suppose that all $J_k$'s
  have a single bottom element $\bot_k$. Then necessarily, $K(y)\cap
  K(z)=\varnothing$, otherwise $\eta(y)\cap\eta(z)=\varnothing$ would not be
  true.  Then it suffices to take $K(x):=K(y)$, $K(x')=\{1,\ldots,l\}\setminus
  K(x)$ and the conditions (i) and (ii) above are satisfied. Conversely, assume
  that there exist some connected component $J_k$ with two bottom elements, say
  $\bot_k$ and $\bot'_k$. Consider $y,z$ such that $\eta(y)=\bot_k$ and
  $\eta(z)=\bot'_k$. Then $(y,z)\in\widetilde{L}$, but due to (\ref{eq:tile}),
  no $x$ can satisfy condition (ii) above.
\end{proof}
\begin{quote}
\textbf{Example 1 (ctd):} Consider $L=2^N$. Then $\widetilde{L}=\Q(N)$. Since
$2^N$ is Boolean, any element $A\subseteq N$ is complemented ($A'=A^c$), and
$2^N(A) = [(\varnothing,\varnothing), (A, A^c)]$. Obviously the conditions of
Theorem \ref{th:biex} are satisfied, thus
\[
\Q(N) = \bigcup_{A\subseteq N}[(\varnothing,\varnothing), (A, A^c)].
\] 
\end{quote}
This important result shows that $\widetilde{L}$ is composed by ``tiles'', all
identical to $L$ (note however that the union is not disjoint). This suggests the
following definition.
\begin{definition}
Let $L$ be a finite distributive lattice, and $\widetilde{L}$ its bipolar
extension. $\widetilde{L}$ is said to be a \emph{regular mosaic} if
$\mathcal{J}(L)$ has all its connected components with a single bottom element.
\end{definition}
There are two important particular cases of regular mosaics:
\begin{enumerate}
\item $L$ is a product of $m$ linear lattices (totally ordered). Then
\[
c(L) = \{(\top_A,\bot_{A^c})\mid A\subseteq \{1,\ldots,m\}\}
\]
where $(\top_A,\bot_{A^c})$ has coordinate number $i$ equal to $\top_i$ if $i\in
A$, and $\bot_i$ otherwise. Also,
$(\top_A,\bot_{A^c})'=(\bot_A,\top_{A^c})$. This case covers Boolean
lattices (case of capacities), and lattices of the form $k^m$, which we will
address in Section \ref{sec:kary}.
\item $\mathcal{J}(L)$ has a single connected component with one bottom element.
  Then $\widetilde{L}$ contains only elements of the form $(y,\bot)$ or
  $(\bot,z)$, i.e., $\widetilde{L}=L(\bot)\cup L(\top)$. 
\end{enumerate}
The following example shows a case where $\tilde{L}$ is not a regular mosaic.
\begin{quote}
\textbf{Example 2:} we consider $L$ and $\mathcal{J}(L)$ given on Figure
\ref{fig:ex1}. Obviously,  $\mathcal{J}(L)$ does not satisfy the condition for
producing a regular mosaic, and as it can be seen on Fig. \ref{fig:ex2}, the
bipolar structure cannot be obtained as a replication of $L$.
\begin{figure}[htb]
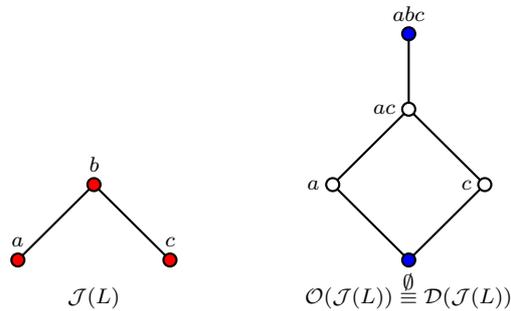

\begin{center}
\psset{unit=0.5cm}
\pspicture(0,0)(4,8)
\psline(2,2)(0,0)
\psline(2,2)(4,0)
\pscircle[fillstyle=solid,fillcolor=red](0,0){0.2}
\pscircle[fillstyle=solid,fillcolor=red](2,2){0.2}
\pscircle[fillstyle=solid,fillcolor=red](4,0){0.2}
\uput[90](0,0){\scriptsize $a$}
\uput[90](2,2){\scriptsize $b$}
\uput[90](4,0){\scriptsize $c$}
\rput(2,-1){\scriptsize $\mathcal{J}(L)$}
\endpspicture
\hspace*{2cm}
\pspicture(0,0)(4,8)
\pspolygon(2,0)(4,2)(2,4)(0,2)
\psline(2,4)(2,6)
\pscircle[fillstyle=solid,fillcolor=blue](2,0){0.2}
\pscircle[fillstyle=solid](0,2){0.2}
\pscircle[fillstyle=solid](4,2){0.2}
\pscircle[fillstyle=solid](2,4){0.2}
\pscircle[fillstyle=solid,fillcolor=blue](2,6){0.2}
\uput[-90](2,0){\scriptsize $\emptyset$}
\uput[90](2,6){\scriptsize $abc$}
\uput[180](0,2){\scriptsize $a$}
\uput[180](4,2){\scriptsize $c$}
\uput[180](2,4){\scriptsize $ac$}
\rput(2,-1){\scriptsize $\mathcal{O}(\mathcal{J}(L))\equiv\mathcal{D}(\mathcal{J}(L))$} 
\endpspicture
\end{center}
\caption{A lattice $L$ and the associated $\mathcal{J}(L)$. In grey, the
  complemented elements}
\label{fig:ex1}
\end{figure}

\begin{figure}[htb]
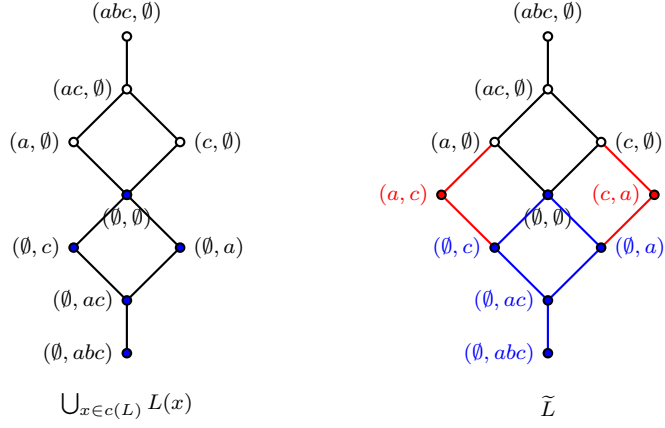

\begin{center}
\psset{unit=0.35cm}
\pspicture(0,-6)(4,8)
\pspolygon(2,0)(4,2)(2,4)(0,2)
\psline(2,4)(2,6)
\pscircle[fillstyle=solid](2,0){0.2}
\pscircle[fillstyle=solid](0,2){0.2}
\pscircle[fillstyle=solid](4,2){0.2}
\pscircle[fillstyle=solid](2,4){0.2}
\pscircle[fillstyle=solid](2,6){0.2}
\uput[-90](2,0){\scriptsize $(\emptyset,\emptyset)$}
\uput[90](2,6){\scriptsize $(abc,\emptyset)$}
\uput[180](0,2){\scriptsize $(a,\emptyset)$}
\uput[0](4,2){\scriptsize $(c,\emptyset)$}
\uput[180](2,4){\scriptsize $(ac,\emptyset)$}
\pspolygon(2,0)(4,-2)(2,-4)(0,-2)
\psline(2,-4)(2,-6)
\pscircle[fillstyle=solid,fillcolor=blue](2,0){0.2}
\pscircle[fillstyle=solid,fillcolor=blue](0,-2){0.2}
\pscircle[fillstyle=solid,fillcolor=blue](4,-2){0.2}
\pscircle[fillstyle=solid,fillcolor=blue](2,-4){0.2}
\pscircle[fillstyle=solid,fillcolor=blue](2,-6){0.2}
\uput[180](2,-6){\scriptsize $(\emptyset,abc)$}
\uput[180](0,-2){\scriptsize $(\emptyset,c)$}
\uput[0](4,-2){\scriptsize $(\emptyset,a)$}
\uput[180](2,-4){\scriptsize $(\emptyset,ac)$}
\rput(2,-8){\scriptsize $\bigcup_{x\in c(L)}L(x)$} 
\endpspicture
\hspace*{4cm}
\psset{unit=0.35cm}
\pspicture(0,-6)(4,8)
\pspolygon(2,0)(4,2)(2,4)(0,2)
\psline(2,4)(2,6)
\psline[linecolor=red](0,2)(-2,0)(0,-2)
\psline[linecolor=red](4,2)(6,0)(4,-2)
\pscircle[fillstyle=solid,fillcolor=red](-2,0){0.2}
\pscircle[fillstyle=solid,fillcolor=red](6,0){0.2}
\uput[180](-2,0){\red\scriptsize $(a,c)$}
\uput[180](6,0){\red\scriptsize $(c,a)$}
\pscircle[fillstyle=solid](2,0){0.2}
\pscircle[fillstyle=solid](0,2){0.2}
\pscircle[fillstyle=solid](4,2){0.2}
\pscircle[fillstyle=solid](2,4){0.2}
\pscircle[fillstyle=solid](2,6){0.2}
\uput[-90](2,0){\scriptsize $(\emptyset,\emptyset)$}
\uput[90](2,6){\scriptsize $(abc,\emptyset)$}
\uput[180](0,2){\scriptsize $(a,\emptyset)$}
\uput[0](4,2){\scriptsize $(c,\emptyset)$}
\uput[180](2,4){\scriptsize $(ac,\emptyset)$}
\pspolygon[linecolor=blue](2,0)(4,-2)(2,-4)(0,-2)
\psline[linecolor=blue](2,-4)(2,-6)
\pscircle[fillstyle=solid,fillcolor=blue](2,0){0.2}
\pscircle[fillstyle=solid,fillcolor=blue](0,-2){0.2}
\pscircle[fillstyle=solid,fillcolor=blue](4,-2){0.2}
\pscircle[fillstyle=solid,fillcolor=blue](2,-4){0.2}
\pscircle[fillstyle=solid,fillcolor=blue](2,-6){0.2}
\uput[180](2,-6){\blue\scriptsize $(\emptyset,abc)$}
\uput[180](0,-2){\blue\scriptsize $(\emptyset,c)$}
\uput[0](4,-2){\blue\scriptsize $(\emptyset,a)$}
\uput[180](2,-4){\blue\scriptsize $(\emptyset,ac)$}
\rput(2,-8){\scriptsize $\widetilde{L}$}
\endpspicture
\end{center}
\caption{Left: bipolar structure computed as a replication of $L$. Right: the
  true bipolar structure}
\label{fig:ex2}
\end{figure}
\end{quote}

\subsection{Bipolar geometric realization}\label{sec:bige}
Since $\widetilde{L}$ is not a distributive lattice, it is not possible to
define its geometric realization in the sense of Def.  \ref{def:geore}. Assuming
that $\widetilde{L}$ is a regular mosaic, we propose the
following definition. 
\begin{definition}\label{def:bipogeore}
Let $\widetilde{L}$ be a regular mosaic, and $x\in c(L)$. We
consider the mappings $\xi_x:\mathcal{J}(L)\rightarrow \{-1,0,1\}$ such that
\begin{enumerate}
\item $|\xi_x|$ is nonincreasing
\item $\xi_x(j) \geq 0$ if $j\in\eta(x)$
\item $\xi_x(j) \leq 0$ if $j\in\eta(x')$.
\end{enumerate}
The set of such functions is denoted by
$\mathcal{D}_x(\mathcal{J}(L))$. Similarly, we introduce
\begin{multline}
\mathcal{C}_x(\mathcal{J}(L)):=\{f_x:\mathcal{J}(L)\rightarrow [-1,1]\text{
  such that }
|f_x| \text{ is nonincreasing},\\
 f_x(j)\geq 0 \text{ if } j\in\eta(x), f_x(j)\leq 0 \text{ if } j\in \eta(x')\}.
\end{multline}
Then the \emph{bipolar geometric realization of $L$} is
\[
\widetilde{|L|}:=\bigcup_{x\in c(L)}\mathcal{C}_x(\mathcal{J}(L)).
\]
\end{definition}
\begin{proposition}
For any $x\in c(L)$,
$\mathcal{D}_x(\mathcal{J}(L))$ is the set of vertices of
$\mathcal{C}_x(\mathcal{J}(L))$.
\end{proposition}
\begin{proof}
It is plain that any $\xi_x$ is a vertex
of $\mathcal{C}_x(\mathcal{J}(L))$. Conversely, assume $f_x$
 is a vertex such that for some $j\in \mathcal{J}(L)$, $f_x(j)=\alpha>0$ (or
 $<0$). Then we define
\[
f^+(j):=f_x(j)+\epsilon, \quad f^-(j):=f_x(j)-\epsilon,
\]
and $f^+=f^-=f_x$ elsewhere, choosing $0<\epsilon<\alpha$ small enough so that
$|f^+|,|f^-|$ remain nonincreasing. Then $f^+,f^-$ belong to
$\mathcal{C}_x(\mathcal{J}(L))$, and $f_x=\frac{1}{2}(f^++f^-)$, which proves
that $f_x$ is not a vertex. 
\end{proof}

\begin{proposition}\label{prop:dlx}
Let $x\in c(L)$. There is a bijection
$\psi_x:\mathcal{D}_x(\mathcal{J}(L))\rightarrow L(x)$ defined by
$\psi_x(\xi):=(y_\xi,z_\xi)$ with
\begin{equation}\label{eq:psi1}
\eta(y_\xi)=\{j\in \mathcal{J}(L)\mid \xi(j)=1\}, \quad \eta(z_\xi)
= \{j\in \mathcal{J}(L)\mid\xi(j)=-1\},
\end{equation}
and the inverse function is defined by $\psi_x^{-1}(y,z):=\xi_{(y,z)}$ with
\begin{equation}\label{eq:psi-1}
\xi_{(y,z)}(j):=\begin{cases}
                                1, & \text{if } j\in \eta(y)\\
                                -1, & \text{if } j\in \eta(z)\\
                                0, & \text{otherwise},
                                \end{cases}
\end{equation}
for any $j\in \mathcal{J}(L)$, or in more compact form
\[
\xi_{(y,z)}= 1_{\eta(y)} -1_{\eta(z)}.
\]
\end{proposition}
\begin{proof}
  Since $|\xi|$ is nonincreasing, $\{j\in \mathcal{J}(L)\mid \xi(j)=1\}$ and
  $\{j\in \mathcal{J}(L)\mid \xi(j)=-1\}$ are downsets. Hence $y_\xi,z_\xi$ are
  well-defined, and by construction $(y_\xi,z_\xi)\in L(x)$.

  Let us show that $|\xi_{(y,z)}|$ is nonincreasing. Assume $\xi_{(y,z)}(j) = 1 $ or $-1$. Then $j\in \eta(y)\cup\eta(z)$. Since these are
  downsets, any $j'\leq j $ belongs also to $\eta(y)\cup\eta(z)$. Assume
  $\xi_{(y,z)}(j) =0$, i.e., $j\not\in \eta(y)\cup\eta(z)$.
  Then $j'\geq j$ cannot belong to $\eta(y)\cup\eta(z)$ since they are downsets,
  hence $\xi_{(y,z)}(j')=0$ .

  Finally, $\psi_x$ is one-to-one because $L$ is distributive, and so is $L(x)$
  (Birkhoff's theorem). 
\end{proof}
\begin{quote}
\textbf{Example 1 (ctd):} Consider $L=2^N$, and some $N^+\subseteq N$,
$N^-:=N\setminus N^+$. Then
\[
\mathcal{D}_{N^+}(N) = \Big\{\xi_{N^+}:N\rightarrow \{-1,0,1\}\text{ such that }
(\xi_{N^+})_{|N^+}\geq 0, \quad (\xi_{N^+})_{|N^-}\leq 0\Big\}.
\] 
Moreover, $\psi_{N^+}(\xi_{N^+}) = (\{j\in N\mid \xi_{N^+}(j)=1\}, \{j\in
  N\mid \xi_{N^+}(j)=-1\})$. 
\end{quote}

 Figure \ref{fig:all} should make things clear for notions
  introduced till this point.
\begin{figure}[htb]
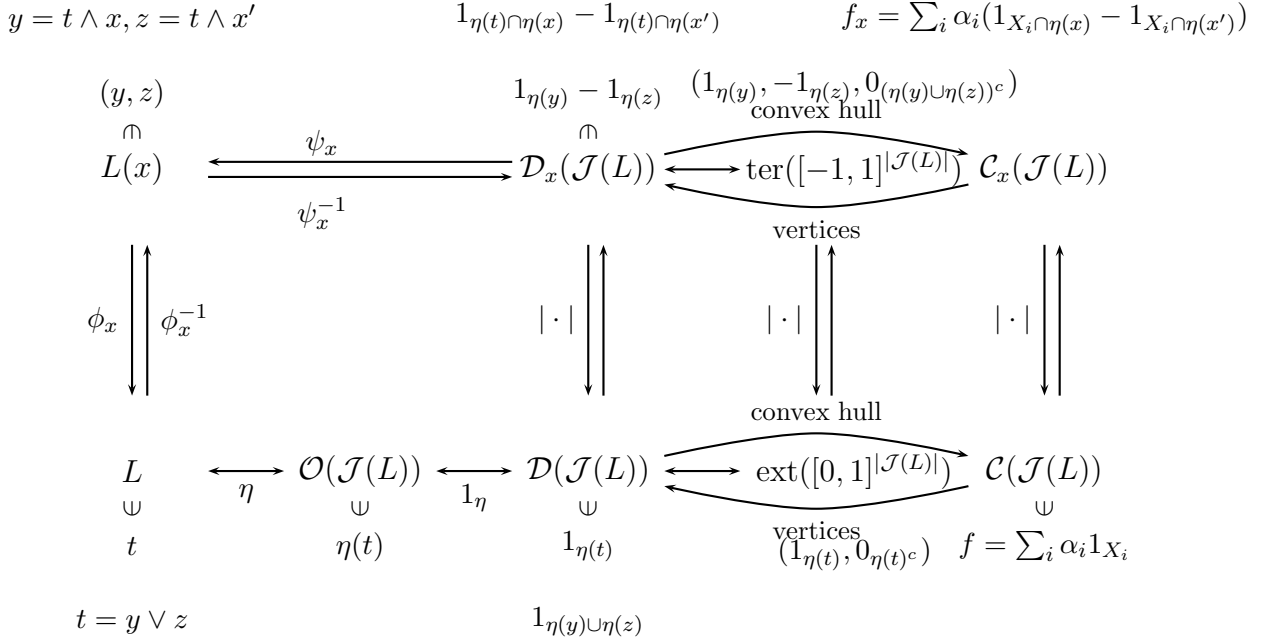

\begin{center}
\psset{unit=1cm}
\pspicture(-2,-2)(14,6)
\psline{<-}(0,1)(0,3)
\psline{->}(0.2,1)(0.2,3)
\psline{<-}(6,1)(6,3)
\psline{->}(6.2,1)(6.2,3)
\psline{<-}(9,1)(9,3)
\psline{->}(9.2,1)(9.2,3)
\psline{<-}(12,1)(12,3)
\psline{->}(12.2,1)(12.2,3)
\psline{<->}(1,0)(2,0)
\psline{<->}(4,0)(5,0)
\psline{<->}(7,0)(8,0)
\pscurve{->}(7,0.2)(9,0.5)(11,0.2)
\pscurve{<-}(7,-0.2)(9,-0.5)(11,-0.2)
\psline{<-}(1,4.1)(5,4.1)
\psline{->}(1,3.9)(5,3.9)
\psline{<->}(7,4)(8,4)
\pscurve{->}(7,4.2)(9,4.5)(11,4.2)
\pscurve{<-}(7,3.8)(9,3.5)(11,3.8)
\rput(0,0){$L$}
\rput(3,0){$\mathcal{O}(\mathcal{J}(L))$}
\rput(6,0){$\mathcal{D}(\mathcal{J}(L))$}
\rput(9.5,0){ext$([0,1]^{|\mathcal{J}(L)|})$}
\rput(12,0){$\mathcal{C}(\mathcal{J}(L))$}
\rput(0,4){$L(x)$}
\rput(6,4){$\mathcal{D}_x(\mathcal{J}(L))$}
\rput(9.5,4){ter$([-1,1]^{|\mathcal{J}(L)|})$}
\rput(12,4){$\mathcal{C}_x(\mathcal{J}(L))$}
\uput[180](0,2){\small $\phi_x$}
\uput[0](0.2,2){\small $\phi_x^{-1}$}
\uput[180](6,2){\small $|\cdot|$}
\uput[180](9,2){\small $|\cdot|$}
\uput[180](12,2){\small $|\cdot|$}
\uput[-90](1.5,0){\small $\eta$}
\uput[-90](4.5,0){\small $1_{\eta}$}
\uput[90](9,0.5){\footnotesize convex hull}
\uput[-90](9,-0.5){\footnotesize vertices}
\uput[90](2.5,4){\small $\psi_x$}
\uput[-90](2.5,3.8){\small $\psi_x^{-1}$}
\uput[90](9,4.5){\footnotesize convex hull}
\uput[-90](9,3.5){\footnotesize vertices}
\rput(0,-0.5){\small \rotateleft{$\in$}}
\rput(0,-1){\small $t$}
\rput(3,-1){\small $\eta(t)$}
\rput(3,-0.5){\small \rotateleft{$\in$}}
\rput(6,-1){\small $1_{\eta(t)}$}
\rput(6,-0.5){\small \rotateleft{$\in$}}
\rput(9.5,-1.1){\small $(1_{\eta(t)},0_{\eta(t)^c})$}
\rput(12,-1){\small $f=\sum_i \alpha_i1_{X_i}$}
\rput(12,-0.5){\small \rotateleft{$\in$}}
\rput(0,-2){\small $t=y\vee z$}
\rput(6,-2){\small $1_{\eta(y)\cup\eta(z)}$}
\rput(0,5){\small $(y,z)$}
\rput(0,4.5){\small \rotateright{$\in$}}
\rput(6,5){\small $1_{\eta(y)}-1_{\eta(z)}$}
\rput(6,4.5){\small \rotateright{$\in$}}
\rput(9.5,5.1){\small $(1_{\eta(y)},-1_{\eta(z)},0_{(\eta(y)\cup\eta(z))^c})$}
\rput(12,6){\small $f_x=\sum_i \alpha_i(1_{X_i\cap \eta(x)}-1_{X_i\cap
    \eta(x')})$} 
\rput(0,6){\small $y=t\wedge x,z=t\wedge x'$}
\rput(6,6){\small $1_{\eta(t)\cap\eta(x)}-1_{\eta(t)\cap\eta(x')}$}
\endpspicture
\caption{Relations among various concepts introduced. $|\cdot|$ indicates
  absolute value, and ter$()$ indicates vectors whose components are $-1$, 1 or
  0.}
\label{fig:all}
\end{center}
\end{figure}
Observe that functions $\xi_x\in\mathcal{D}_x(\mathcal{J}(L))$ corresponds to a
subset of points of $[-1,1]^{|\mathcal{J}(L)|}$ of the form
$(1_A,(-1)_B,0_{(A\cup B)^c})$, with $A\subseteq \eta(x)$ and $B\subseteq
\eta(x')$, and that $\mathcal{C}_x(\mathcal{J}(L))$ is the convex hull of these
points.

\medskip

We end this section by addressing the natural triangulation of the bipolar
geometric realization. Let us consider some $f$ in
$\mathcal{C}(\mathcal{J}(L))$, assuming $f=\sum_{i=0}^p\alpha_i1_{X_i}$, with
$1_{X_0},\ldots,1_{X_p}$ forming a chain in $\mathcal{D}(\mathcal{J}(L))$. Given
$x\in c(L)$, let us define the corresponding $f_x$ in
$\mathcal{C}_x(\mathcal{J}(L))$ as follows:
\begin{align*}
f_x & :=\sum_{i=0}^p \alpha_i\psi_x^{-1}(\phi_x^{-1}(\eta^{-1}(X_i)))\\
 & = \sum_{i=0}^p \alpha_i(1_{X_i\cap\eta(x)} - 1_{X_i\cap\eta(x')}).
\end{align*}
Explicitely, this gives, for any $j\in \mathcal{J}(L)$:
\[
f_x(j) = \begin{cases}
  \sum_{i\mid j\in X_i}\alpha_i, & \text{if } j\in\eta(x)\\
  -\sum_{i\mid j\in X_i}\alpha_i, & \text{if } j\in\eta(x').
  \end{cases}
\]
Hence $|f_x|$ takes value 1 on $X_0$, $1-\alpha_0$ on $X_1\setminus X_0$, etc.,
and is nonincreasing.

Remark that $|f_x|=f$ if $f\in\mathcal{C}(\mathcal{J}(L))$, and $|f|_x=f$ if
$f\in\mathcal{C}_x(\mathcal{J}(L))$.

\subsection{Natural interpolation on bipolar structures}\label{sec:nabi}
Again we suppose that $\widetilde{L}$ is a regular mosaic.
Assume $F:\bigcup_{x\in c(L)}\mathcal{D}_x(\mathcal{J}(L))\rightarrow \mathbb{R}$
is given. We want to define the extension $\overline{F}$ of this functional on the
bipolar geometric realization $\widetilde{|L|}$. 

Let us take $f\in\widetilde{|L|}=\bigcup_{x\in
  c(L)}\mathcal{C}_x(\mathcal{J}(L))$.  First, we must choose $x\in c(L)$ such
  that $f$ belongs to $\mathcal{C}_x(\mathcal{J}(L))$ ($x$ is not unique in
  general since in the definition of $\widetilde{|L|}$ the union is not
  disjoint (see Def. \ref{def:bipogeore})). Defining
\[
\mathcal{J}(L)^+:=\{j\in \mathcal{J}(L)\mid f(j)\geq 0\}, \quad
\mathcal{J}(L)^-:= \mathcal{J}(L)\setminus \mathcal{J}(L)^+,
\] 
it suffices to take $x,x'$ defined by
\[
\eta(x):=\bigcup_{k\in K}J_k, \quad \eta(x'):= \mathcal{J}(L)\setminus \eta(x)
\]
with $K$ the smallest one such that $\mathcal{J}(L)^+\subseteq \bigcup_{k\in
  K}J_k$ (using notations of proof of Theorem~\ref{th:biex}).
Now, consider $|f|$, which belongs to $\mathcal{C}(\mathcal{J}(L))$, and its
expression using the natural triangulation:
\[
|f|=\sum_{i=0}^p \alpha_i 1_{X_i} 
\]
with $1_{X_0},\ldots,1_{X_p}$ a chain in $\mathcal{D}(\mathcal{J}(L))$. Then we
have $|f|_x=f$, and we propose the following definition.
\begin{definition}\label{def:bina}
Assume $\widetilde{L}$ is a regular mosaic. 
For any functional $F:\bigcup_{x\in c(L)}\mathcal{D}_x(\mathcal{J}(L))\rightarrow
\mathbb{R}$, its
\emph{natural extension to the bipolar geometric realization} of $\widetilde{L}$
is defined by:
\[
\overline{F}(f) :=\sum_{i=0}^p\alpha_i F_x(1_{X_i})
\]
for all $f\in\mathcal{C}_x(\mathcal{J}(L))$, letting $|f|:=\sum_{i=0}^p \alpha_i
1_{X_i}$ for some chain $\{1_{X_0}<1_{X_1}<\cdots<1_{X_p}\}$ in
$\mathcal{D}(\mathcal{J}(L))$, and $F_x:\mathcal{D}(\mathcal{J}(L))\rightarrow
\mathbb{R}$ defined by:
\[
F_x(1_{X_i}):=F(1_{X_i\cap\eta(x)} - 1_{X_i\cap\eta(x')}).
\]
\end{definition}
\begin{quote}
\textbf{Example (end):} Let us take once more $L=2^N$. 
For a given  $f$, we define  $N^+:=\{j\in N\mid f(j)\geq 0\}$ and
$N^-:=N\setminus N^+$, we have:
\[
\overline{F}(f) = \sum_{i=1}^n\alpha_iF_{N^+}(1_{X_i})
= \sum_{i=1}^n\big[|f(\pi(i))| - |f(\pi(i+1))|\big]F(1_{X_i\cap
N^+}-1_{X_i\cap N^-}),
\]
where we have used (\ref{eq:alpha}).  Putting $v(A,B):=F(1_A-1_B)$, we recognize
the Choquet integral for bicapacities (see Definition \ref{def:bicho}). $\Box$
\end{quote}

\begin{remark}
Definition \ref{def:bina} can be written equivalently as $\overline{F}(f) =
\overline{F_x}(|f|)$, making clear the relation between the functional on $L$
and on $\widetilde{L}$. 
\end{remark}

Lastly, we address the problem of expressing $\overline{F}$ in terms of the
M\"obius transform of $F$, using Prop. \ref{prop:mob}. For this purpose, it is
better to turn a given functional $F$ on $\bigcup_{x\in
c(L)}\mathcal{D}_x(\mathcal{J}(L))$ into its equivalent form $\widetilde{F}$
defined on $\widetilde{L}$, thanks to the mappings $\psi_x$, $x\in c(L)$. Doing
so, we can use Prop. \ref{prop:mob} and (\ref{eq:bms}), and get the M\"obius
transform of $\widetilde{F}$, which we denote by $\widetilde{m}$:
\[
\widetilde{m}(x,y)=\sum_{\substack{(z,t)\leq
    (x,y)\\(z,t)\in\widetilde{L}}}\widetilde{F}(z,t)\mu_L(z,x)\mu_L(t,y),
    \quad\forall (x,y)\in \widetilde{L}.
\]
We need the following result, which is a generalization of (\ref{eq:musu}).
\begin{lemma}
Let $f\in\mathcal{C}(\mathcal{J}(L))$ and
$F:\mathcal{D}(\mathcal{J}(L))\rightarrow\{0,1\}$ being nondecreasing and 0-1
valued. Then
\[
\overline{F}(f) = \bigvee_{\substack{T\subseteq
    \mathcal{J}(L)\\F(1_T)=1}}\bigwedge_{j\in T} f(j).
\]
\end{lemma}
\begin{proof}
(adaptation from \cite{musu93}) Using notations of Proposition \ref{prop:rem2}, define
  $i_0\in\mathcal{J}(L)$  such that 
\[
f(\pi(i_0))=\bigvee_{\substack{T\subseteq
    \mathcal{J}(L)\\F(1_T)=1}}\bigwedge_{j\in T} f(j).
\]
Assume for simplicity that $f(\pi(1))>f(\pi(2))>\cdots>f(\pi(n))$, and let us
show that
\[
F(1_{\{\pi(1),\ldots,\pi(i)\}})=\begin{cases}
    1, &\text{if } i\geq i_0\\
    0, &\text{else.}
    \end{cases}
\]
Assume $i\geq i_0$. Then for any $T\subseteq\mathcal{J}(L)$ such that
$F(1_T)=1$, we have $f(\pi(i))\leq\wedge_{j\in T}f(j)$. This inequality implies
that $T\subseteq\{\pi(1),\ldots,\pi(i)\}$, and hence by monotonicity of $F$,
we get $F(1_{\{\pi(1),\ldots,\pi(i)\}})=1$. Now suppose $i<i_0$. If
$F(1_{\pi(1),\ldots,\pi(i)})=1$, it follows that $f(\pi(i))>f(\pi(i_0))\geq
\wedge_{j=1}^if(\pi(j))= f(\pi(i))$, a contradiction. Hence
$F(1_{\{\pi(1),\ldots,\pi(i)\}})=0$. 

Using this result in (\ref{eq:natin}) gives the desired result.
\end{proof}

The following is a generalization of (\ref{eq:2}).
\begin{proposition}
With the above notations, for any $f\in\widetilde{|L|}$ and any $F$ on
$\bigcup_{x\in c(L)}\mathcal{D}_x(\mathcal{J}(L))$, the following holds:
\[
\overline{F}(f) = \sum_{(s,t)\in\widetilde{L}}\widetilde{m}(s,t)\Big[\bigwedge_{j\in\eta(s)}f^+(j)\wedge\bigwedge_{j\in\eta(t)}f^-(j)\Big],
\]
with $f^+=f\vee 0$, $f^-=(-f)^+$.
\end{proposition}
\begin{proof}
Taking $\widetilde{F}:=u_{(s,t)}$ given by (\ref{eq:ug}), we have
by Definition \ref{def:bina}
\[
\overline{F}(f) = \overline{F_x}(|f|)
\] 
with $\widetilde{F_x}(y)=u_{(s,t)}(y\wedge x,y\wedge x')$ a nondecreasing 0-1
valued function, with value 1 iff $y\wedge x\geq s$ and $z\wedge x'\geq
t$. Since $L$ is distributive, this condition writes [$\eta(y)\cap
\eta(x)\supseteq \eta(s)$ and $\eta(z)\cap \eta(x')\supseteq \eta(t)$], which in
turn is equivalent to [$\eta(y)\supseteq \eta(s)\cup\eta(t)$ and
$\eta(s)\subseteq \eta(x)$ and $\eta(t)\subseteq \eta(x')$] since $x,x'$ are
complemented. Hence, applying the above lemma, we get:
\begin{align*}
\overline{F_x}(|f|)  & = \bigvee_{\substack{y\geq s\vee t\\s\leq x\\t\leq
                    x'}}\bigwedge_{j\in \eta(y)} |f(j)|\\ 
                    & = \bigwedge_{j\in\eta(s)}f^+(j)\wedge\bigwedge_{j\in
                    \eta(t)}f^-(j).
\end{align*}
Using linearity of $\overline{F_x}$ versus $F_x$ (see Proposition
\ref{prop:rem3} (ii)) and the decomposition of any $\widetilde{F}$ in the basis
of functions $u_{(x,y)}$, the result is proved.
\end{proof}

\section{Application: $k$-ary bicapacities and Choquet integral}\label{sec:kary}
This section is dedicated to the study of the lattice $L:=k^n$.  We set
$N:=\{1,\ldots,n\}$. Elements of $L$ are thus vectors in
$\{0,1,\ldots,k-1\}^n$. For commodity $(l_A,l'_{-A})$ denotes the element $t$ of
$L$ with $t_i=l$ if $i\in A$ and $l'$ otherwise, and we put $M:=n(k-1)$.

We begin by giving a motivation of this study rooted in multicriteria decision
making.

\subsection{Multicriteria aggregation with reference levels}\label{Sinterpret}
Let us consider $N$ as the set of criteria. In the terminology of multicriteria
decision making, an \emph{act} or \emph{option} is a mapping $f:N\rightarrow
\mathbb{R}$, and $f(i)$ is the \emph{score} of option $f$ on criteria $i$. We
may introduce reference levels for scores, and be interested into the overall
score of an option taking values only in the set of reference values (such options
are called \emph{pure}, or \emph{prototypical}). Since these options are
prototypical, the decision maker is able to assess their overall scores. The question
arises then to compute the overall score of an option being not pure. Using our
framework, there are basically two ways of answering this question. We put
$L:=k^n$, where $k$ is the number of reference levels, labelled
$\{0,1,\ldots,k-1\}$. Observe that join-irreducible elements are of the form
$(l_i,0_{-i})$, for any $l\in\{1,\ldots,k-1\}$ and $i\in N$ (see below).

\medskip

The first way is to say that non-pure options belong to a level only to some
degree that can be different from the complete membership and the complete
non-membership.  Thus, as in Fuzzy Set Theory \cite{zad65}, a membership degree
is associated to each level and each criterion, i.e., to each join-irreducible
element.  From a knowledge of these degrees, it is possible to interpolate
between the values known for pure options.

More precisely, an option is an element of $\mathcal{C}(\mathcal{J}(L))$.  A
degree in $[0,1]$ is thus associated to all join-irreducible elements.  It can
be interpreted as a membership degree to the class of levels lower or equal to
the join-irreducible element.  Hence if an option belongs at a given degree
$\delta$ to a join-irreducible element, it necessarily belongs to a degree
greater or equal to $\delta$ to less preferred join-irreducible elements.  This
explains why options shall be non-increasing functions on $\mathcal{J}(L)$.

\medskip

The second way is to map the lattice onto a subset of $\mathbb{R}^n$ such that
the Pareto order on $\mathbb{R}^n$ corresponds precisely to the order relation
on the lattice $k^n$. For this, we map each reference level on $\mathbb{R}$:
$\rho_0<\cdots<\rho_{k-1}$, which represent the score assigned to each level.
The lattice corresponds to the nodes of a rectangular mesh in $\mathbb{R}^n$ composed
of the $k$ reference levels for each criterion. The generalized capacity gives
the value associated to these nodes (i.e., the pure options). The non-pure
options are any point inside the mesh. The problem becomes thus an interpolation
problem in $\mathbb{R}^n$.

Consider thus an option $x\in [\rho_0,\rho_{k-1}]^n$ and a generalized capacity 
$F:\mathcal{D}(\mathcal{J}(L))\rightarrow\mathbb{R}$. 
Let $I(x)\in \{1,\ldots,k\}^n$ such that for any $i\in N$
\[ \rho_{I_i(x)-1} \leq x_i \leq  \rho_{I_i(x)}.
\]
Define $\Phi:[\rho_0,\rho_{k-1}]^n \rightarrow [0,1]^n$ as
\[ \Phi_i(x) := \frac{x_i-\rho_{I_i(x)-1}}{\rho_{I_i(x)} - \rho_{I_i(x)-1}}.
\]
Define a capacity $v_x$ on $N$ by
\[ v_x(S) := F\Big(1_{\bigcup_{i\in N} \big\{(1_i,0_{-i}),\ldots,((I_i(x)-1)_i,0_{-i})\big\}
  \cup \bigcup_{i\in S}((I_i(x))_i,0_{-i}) }\Big)
\]
for all $S\subseteq N$.  It corresponds to the value on the $2^n$ nodes of the
mesh just around $x$.  One may have $v_x(\emptyset)\not=0$. Let
$v'_x(S)=v_x(S)-v_x(\emptyset)$ and $\eta$ a permutation on $N$ such that
$\Phi_{\eta(1)}(x) \geq \cdots\geq \Phi_{\eta(n)}(x)$.
\begin{align}
   v_x(\emptyset) + C_{v'_x}(\Phi(x))  &
  = v_x(\emptyset) + \sum_{i=1}^n (\Phi_{\eta(i)}(x)-\Phi_{\eta(i+1)}(x))
   \: (v_x(\{\eta(1),\ldots,\eta(i)\}) - v_x(\emptyset)) \nonumber \\
 & = v_x(\emptyset) \: (1-\Phi_{\eta(1)}(x)) + \sum_{i=1}^n (\Phi_{\eta(i)}(x)-\Phi_{\eta(i+1)}(x))
   \: v_x(\{\eta(1),\ldots,\eta(i)\}).  \label{E1}
\end{align}

\subsection{The unipolar case}\label{sec:unik}
The set $\mathcal{J}(L)$ of join-irreducible elements is
\[
\mathcal{J}(L) = \Big\{(l_i,0_{-i}) \mid l\in\{1,\ldots,k-1\}, i\in N\Big\}.
\]
It is a poset with $n$ connected components, each of them being the linear
lattice \linebreak $\{1,\ldots k-1\}$.

Let us consider $f$ an element of $\mathcal{C}(\mathcal{J}(L))$. We set for
commodity $f_i^l:=f(l_i,0_{-i})$. The natural triangulation of
$\mathcal{C}(\mathcal{J}(L))$ is done through chains in
$\mathcal{D}(\mathcal{J}(L))$, and maximal chains correspond to some
permutations on $\mathcal{J}(L)$ (see Proposition \ref{prop:rem2}). For
commodity to each permutation $\pi:\{1,\ldots,M\}\rightarrow
\mathcal{J}(L)$ we assign two functions $\lambda:\{1,\ldots,M\}\rightarrow
\{1,\ldots,k-1\}$ and $\theta:\{1,\ldots,M\}\rightarrow \{1,\ldots,n\}$ such
that $\pi(i)= (\lambda(i)_{\theta(i)},0_{-\theta(i)})$, for all
$i\in\{1,\ldots,M\}$.

Applying Proposition \ref{prop:rem2} again, we know that for any element $f$ of
a simplex of $\mathcal{C}(\mathcal{J}(L))$ corresponding to a permutation $\pi$
on $\mathcal{J}(L)$, we have
\[ f_{\theta(1)}^{\lambda(1)} \geq f_{\theta(2)}^{\lambda(2)} \geq \cdots \geq f_{\theta(M)}^{\lambda(M)} 
\]
and 
\[ f(l_p,0_{-p}) = \sum_{i\in \{1,\ldots,M\}}
   \alpha_i 1_{X_i}(l_p,0_{-p})
\]
where
$X_i:=\{(\lambda(1)_{\theta(1)},0_{-\theta(1)}),\ldots,(\lambda(i)_{\theta(i)},0_{-\theta(i)})\}$,
$\alpha_i = f_{\theta(i)}^{\lambda(i)}-f_{\theta(i+1)}^{\lambda(i+1)}$ for
\linebreak $i\in \{1,\ldots,M-1\}$, and $\alpha_{M}=f_{\theta(M)}^{\lambda(M)}$.

A \emph{$k$-ary capacity} is a function
$F:\mathcal{D}(\mathcal{J}(L))\rightarrow \mathbb{R}$. Applying Proposition
\ref{prop:rem3} the natural extension of $f\in \mathcal{C}(\mathcal{J}(L))$
w.r.t. $F$ is
\[ \overline{F}(f) = \sum_{i=1}^{M} \Big[f_{\theta(i)}^{\lambda(i)}-f_{\theta(i+1)}^{\lambda(i+1)}\Big]
  \times
  F\Big(1_{\{(\lambda(1)_{\theta(1)},0_{-\theta(1)}),\ldots,(\lambda(i)_{\theta(i)},0_{-\theta(i)})\}}\Big),
\]
with $f^{\lambda(M+1)}_{\theta(M+1)}:=0$.
This could be considered as the Choquet integral of $f$ w.r.t $F$.

To recover the interpolation formula (\ref{E1}) of Section \ref{Sinterpret}, we
consider a particular class of elements $f$ in $\mathcal{C}(\mathcal{J}(L))$
satisfying for all $i\in N$
\begin{align*}
& f_i^1 = \cdots = f_i^{J_i(f)-1} = 1\\
& f_i^{J_i(f)} = z_i\\
& f_i^{J_i(f)+1} = \cdots = f_i^{k-1} = 0,
\end{align*}
for some given integers $J_1(f),\ldots,J_n(f)$ in
$\{1,\ldots,k-1\}$, and real numbers $z_1,\ldots,z_n\in[0,1]$. 

Let us denote by $\sigma$ a permutation on $N$ such that
\[ z_{\sigma(1)}\geq \cdots \geq z_{\sigma(n)}.
\]
Remark that $f$ belongs to all $M$-dimensional simplices of
$\mathcal{C}(\mathcal{J}(L))$ whose corresponding permutation satisfy:
\begin{align*}
\forall i\in \{1,\ldots,q^f\} , & \quad f_{\theta(i)}^{\lambda(i)}=1 \\
\forall i\in \{q^f+1,\ldots,q^f+n\}, &  \quad f_{\theta(i)}^{\lambda(i)}=z_{\sigma(i-q^f)} \\
\forall i\in \{q^f+n+1,\ldots M\}, & \quad  f_{\theta(i)}^{\lambda(i)}=0
\end{align*}
where $q^f=\sum_{i\in N} (J_i(f)-1)$. Hence, $f$ belongs to the interior of a
$n$-dimensional simplex corresponding to the chain
\[
 1_{X_{q^f}} < 1_{X_{q^f}\cup\{((J_{\sigma(1)}(f))_{\sigma(1)},0_{-\sigma(1)})\}} <\cdots < 1_{X_{q^f}\cup\{((J_{\sigma(1)}(f))_{\sigma(1)},0_{-\sigma(1)}),\ldots,((J_{\sigma(n)}(f))_{\sigma(n)},0_{-\sigma(n)})\}},
\]
with $X_{q^f}:=\{(l_i,0_{-i})\mid 1\leq l_i\leq J_i(f)-1\}$. Then
\begin{equation}
  \overline{F}(f) = (1-z_{\sigma(1)}) \: F(1_{X_{q^f}})
  + \sum_{i=1}^n (z_{\sigma(i)}-z_{\sigma(i+1)}) \: F(1_{X_{q^f}\cup 
  \{ ((J_{\sigma(1)}(f))_{\sigma(1)},0_{-\sigma(1)}),\ldots,
     ((J_{\sigma(i)}(f))_{\sigma(i)},0_{-\sigma(i)}) \} })
\label{E2}
\end{equation}
with $z_{\sigma(n+1)}:=0$.
Let $x\in [\rho_0,\rho_{k-1}]^n$ defined by
\[ x_i := \rho_{J_i(f)-1} + (\rho_{J_i(f)}-\rho_{J_i(f)-1})\times z_i \ .
\]
Then expression (\ref{E1}) and (\ref{E2}) lead to exactly the same value since
$J(f)=I(x)$, $\sigma=\eta$, $z=\Phi(x)$ and
\[ v_x(S) := F(1_{X_{q^f} \: \cup \: \bigcup_{i\in S} ((J_i(f))_i,0_{-i})}) \ .
\]

\subsection{The bipolar case}\label{sec:bipok}
The bipolarization of $L$ is
\[ \tilde{L} = \Big\{ (x,y)\in k^n \times k^n \mid \forall i\in N,\quad x_i\not=0  \Rightarrow y_i=0,
 \mbox{ and } y_i\not=0 \Rightarrow x_i=0 \Big\}.
\]
Moreover, the set of complemented elements is
\[ c(L) = \Big\{ ((k-1)_A,0_{-A}) \mid  A\subseteq N\Big\},
\]
and $((k-1)_A,0_{- A})'=(0_A,(k-1)_{- A})$. Note that $\tilde{L}$ is a regular
mosaic, hence Theorem~\ref{th:biex} applies and $\tilde{L}$ is the union of all $L(x)$, with $x\in c(L)$, and
\[ L((k-1)_A,0_{-A}) = 
  \Big\{\big((x_A,0_{-A}),(0_A,y_{-A})\big)\mid
  x\in\{0,\ldots,k-1\}^{|A|},y\in\{0,\ldots,k-1\}^{n-|A|}\Big\}. 
\]
Let $f\in \mathcal{C}_{((k-1)_A,0_{-A})} (\mathcal{J}(L))$, and
$f_i^l:=f(l_i,0_{-i})$. We have $f_i^l \geq 0$ if $i\in A$ and
$f_i^l \leq 0$ if $i\not\in A$.

We consider a simplex of  $\mathcal{C}_{((k-1)_A,0_{-A})} (\mathcal{J}(L))$
containing $f$, whose corresponding permutation is
$\pi:\{1,\ldots,M\}\rightarrow \mathcal{J}(L)$, and we define as in Section
\ref{sec:unik} the functions $\lambda:\{1,\ldots,M\}\rightarrow
\{1,\ldots,k-1\}$, and $\theta:\{1,\ldots,M\}\rightarrow \{1,\ldots,n\}$.
Then
\[ \Big|f_{\theta(1)}^{\lambda(1)}\Big| \geq \Big|f_{\theta(2)}^{\lambda(2)}\Big| \geq \cdots 
   \geq \Big|f_{\theta(M)}^{\lambda(M)}\Big|
\]
and
\[ \big|f(l_p,0_{-p})\big| = \sum_{i\in \{1,\ldots,M\}}
   \alpha_i \: 1_{X_i}(l_p,0_{-p})
\]
where $X_i:=\{(\lambda(1)_{\theta(1)},0_{-\theta(1)}),\ldots,(\lambda(i)_{\theta(i)},0_{-\theta(i)})\}$, $\alpha_i =
\Big|f_{\theta(i)}^{\lambda(i)}\Big|-\Big|f_{\theta(i+1)}^{\lambda(i+1)}\Big|$
for \linebreak $i\in \{1,\ldots,M-1\}$, and
$\alpha_{M}=\Big|f_{\theta(M)}^{\lambda(M)}\Big|$.

A \emph{$k$-ary bicapacity} is a function 
$F:\cup_{A \subseteq N} \mathcal{D}_{((k-1)_A,0_{-A})}(\mathcal{J}(L))\rightarrow \mathbb{R}$. 

The natural extension $\overline{F}(f)$ is:
\begin{align*}
  \overline{F}(f) = & \sum_{i=1}^{M} \Big(\Big|f_{\theta(i)}^{\lambda(i)}\Big|-\Big|f_{\theta(i+1)}^{\lambda(i+1)}\Big|\Big)
  \times  \\
 & \ \ \ \  \ \ \ \ F\Big(1_{\bigcup_{\substack{q\in\{1,\ldots,i\}\\
  \theta(q)\in A}} (\lambda(q)_{\theta(q)},0_{-\theta(q)})} 
  -1_{\bigcup_{\substack{q\in\{1,\ldots,i\}\\
  \theta(q)\not\in A}} (\lambda(q)_{\theta(q)},0_{-\theta(q)})} \Big),
\end{align*}
with $f^{\lambda(M+1)}_{\theta(M+1)}:=0$.
As before, this could be considered as the Choquet integral of $f$ w.r.t. $F$.  

We consider now a particular class of elements $f$ in $\mathcal{C}(\mathcal{J}(L))$
satisfying for all $i\in N$
\begin{align*}
& |f_i^1| = \cdots = |f_i^{J_i(f)-1}| = 1\\
& |f_i^{J_i(f)}| = z_i\\
& |f_i^{J_i(f)+1}| = \cdots = |f_i^{k-1}| = 0,
\end{align*}
for some given integers $J_1(f),\ldots,J_n(f)$ in
$\{1,\ldots,k-1\}$, and real numbers $z_1,\ldots,z_n\in[0,1]$. 
Let us denote by $\sigma$ a permutation on $N$ such that
\[ z_{\sigma(1)}\geq \cdots \geq z_{\sigma(n)}.
\]
Remark that $f$ belongs to all $M$-dimensional simplices of
$\mathcal{C}(\mathcal{J}(L))$ whose corresponding permutation satisfy:
\begin{align*}
\forall i\in \{1,\ldots,q^f\} , & \quad f_{\theta(i)}^{\lambda(i)}=1\text{ if }
\theta(i)\in A, \text{ and } -1 \text{ otherwise} \\
\forall i\in \{q^f+1,\ldots,q^f+n\}, &  \quad
f_{\theta(i)}^{\lambda(i)}=z_{\sigma(i-q^f)} \text{ if } \theta(i)\in A, \text{
  and } -z_{\sigma(i-q^f)} \text{ otherwise}\\
\forall i\in \{q^f+n+1,\ldots M\}, & \quad  f_{\theta(i)}^{\lambda(i)}=0
\end{align*}
where $q^f=\sum_{i\in N} (J_i(f)-1)$. Then
\begin{equation}
  \overline{F}(f) = (1-z_{\sigma(1)}) V(\varnothing) 
  + \sum_{i=1}^n (z_{\sigma(i)}-z_{\sigma(i+1)})  
   V( \{\sigma(1),\ldots,\sigma(i)\}), 
\label{E3}
\end{equation}
with $z_{\sigma(n+1)}:=0$, and
\[ V(S) := F\Big(1_{\big(X_{q^f} \cup \bigcup_{i\in S} ((J_i(f))_i,0_{-i})\big)
  \cap \mathcal{J}(L)^+} 
  -1_{\big(X_{q^f} \cup \bigcup_{i\in S} ((J_i(f))_i,0_{-i})\big) \cap
   \mathcal{J}(L)^-}\Big),
\]
with $X_{q^f}:=\{(l_i,0_{-i})\mid 1\leq l_i\leq J_i(f)-1\}$.

The positive part of the scale is represented by the positive levels
$\rho_0,\ldots\rho_{k-1}$.  The negative part of the scale is represented by the
negative levels $\rho_{-k+1},\ldots,\rho_0$.  Hence $\rho_0=0$ is the neutral
element demarcarting between attractive and repulsive values.  Let $x\in
[\rho_{-k+1},\rho_{k-1}]^n$ defined by
\[ x_i := \left\{ \begin{array}{l}
  \rho_{J_i(f)-1} + (\rho_{J_i(f)}-\rho_{J_i(f)-1})\times z_i  \ \mbox{ if } i\in A \\
  \rho_{-J_i(f)+1} + (\rho_{-J_i(f)}-\rho_{-J_i(f)+1})\times z_i  \ \mbox{ if } i\not\in A
 \end{array} \right.
\]
Then proceeding as in Section \ref{Sinterpret}, it is easy to see that (\ref{E3})
corresponds exactly to the Choquet integral for $k$-ary bicapacities defined in
\cite{grla03b}.

\subsection{Example}
We end this section by illustrating the above results with $k=3$ and
$n=2$. Clearly, the case $k=2$ was already well described (capacities and
bicapacities). 

Elements in $L:=3^2$ are denoted by pairs $(l_1,l_2)$, with $l_i\in\{0,1,2\}$,
$i=1,2$. We have four join-irreducible elements $(1,0),(2,0),(0,1),(0,2)$. Let
us consider the following function $f$ in $\mathcal{C}(\mathcal{J}(L))$:
\[
f(1,0)=0.5, \quad f(2,0)=0.1, \quad f(0,1)=0.3,\quad f(0,2)=0.2.
\]
Note that $f$ is indeed nonincreasing on $\mathcal{J}(L)$. The associated
permutation is
\[
\pi(1)=(1,0),\quad \pi(2)=(0,1),\quad \pi(3)=(0,2),\quad \pi(4)=(2,0),
\]
and the corresponding maximal chain is (expressed in $L$ for simplicity)
\[
(0,0)<(1,0)<(1,1)<(1,2)<(2,2).
\]
(see Fig. \ref{fig:ex3} for a diagram of $L$, and the maximal chain in bold)
Supposing $F$ being defined on $L$, the Choquet integral of $f$ w.r.t. $F$ is
given by
\begin{multline*}
\overline{F}(f)=[f(1,0)-f(0,1)]F(1,0) + [f(0,1)-f(0,2)]F(1,1)\\
 + [f(0,2)-f(2,0)]F(1,2) + f(2,0)F(2,2).
\end{multline*}

Let us turn to the bipolar case. To avoid a heavy notation, elements of
$\tilde{L}$ are denoted by $(ij,kl)$ instead of $((i,j),(k,l))$. The set of
complemented elements together with their complemented elements is
\begin{align*}
A=\emptyset: \qquad (0,0) & \leftrightarrow (2,2)\\
A=\{1\}: \qquad (2,0) & \leftrightarrow (0,2)\\
A=\{2\}: \qquad (0,2) & \leftrightarrow (2,0)\\
A=\{1,2\}: \qquad (2,2) & \leftrightarrow (0,0)
\end{align*} 
Then $\tilde{L}=L(0,0)\cup L(2,0)\cup L(0,2)\cup L(2,2)$, with
\begin{align*}
L(0,0) & = \{(00,kl)\mid k,l\in\{0,1,2\}\}\\
L(2,0) & = \{(i0,0l)\mid i,l\in\{0,1,2\}\}\\
L(0,2) & = \{(0j,k0)\mid j,k\in\{0,1,2\}\}\\
L(2,2) & = \{(ij,00)\mid i,j\in\{0,1,2\}\}.
\end{align*}
Consider the function $f$ defined by
\[
f(1,0)=0.5, \quad f(2,0)=0.1, \quad f(0,1)=-0.3,\quad f(0,2)=-0.2.
\]
Then $A=\{1\}$, $f\in\mathcal{C}_{(2,0)}(\mathcal{J}(L))$, and the permutation $\pi$ is the
same as above. Now, assuming $F$ defined on $\tilde{L}$ is given,
\begin{multline*}
\overline{F}(f)=[|f(1,0)|-|f(0,1)|]F(10,00) + [|f(0,1)|-|f(0,2)|]F(10,01)\\
 + [|f(0,2)|-|f(2,0)|]F(10,02) + |f(2,0)|F(20,02).
\end{multline*}
Fig. \ref{fig:ex3} shows the bipolar extension $\tilde{L}$, the part $L(2,0)$ used
for $f$ is in grey, and the maximal chain is in bold.
\begin{figure}[htb]
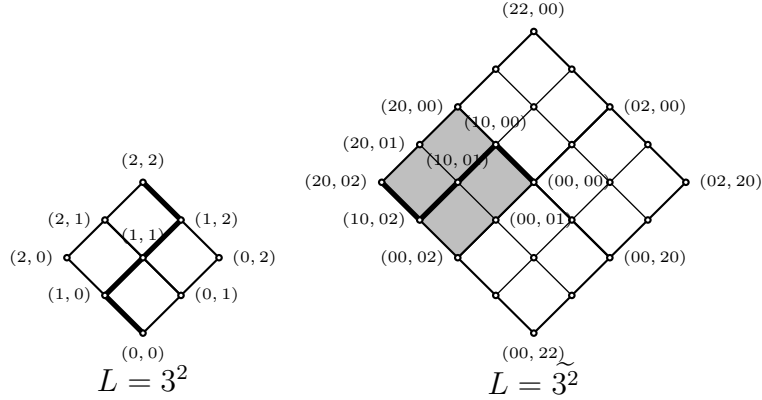

\begin{center}
\psset{unit=0.5cm}
\pspicture(0,0)(4,5)
\pspolygon(2,2)(1,3)(2,4)(3,3)
\pspolygon(2,2)(1,3)(0,2)(1,1)
\pspolygon(2,2)(1,1)(2,0)(3,1)
\pspolygon(2,2)(3,3)(4,2)(3,1)
\psline[linewidth=2pt](2,0)(1,1)(3,3)(2,4)
\pscircle[fillstyle=solid](2,2){0.1}
\pscircle[fillstyle=solid](1,1){0.1}
\pscircle[fillstyle=solid](2,0){0.1}
\pscircle[fillstyle=solid](3,1){0.1}
\pscircle[fillstyle=solid](0,2){0.1}
\pscircle[fillstyle=solid](4,2){0.1}
\pscircle[fillstyle=solid](1,3){0.1}
\pscircle[fillstyle=solid](3,3){0.1}
\pscircle[fillstyle=solid](2,4){0.1}
\uput[90](2,2){\tiny $(1,1)$}
\uput[180](1,1){\tiny $(1,0)$}
\uput[-90](2,0){\tiny $(0,0)$}
\uput[0](3,1){\tiny $(0,1)$}
\uput[180](0,2){\tiny $(2,0)$}
\uput[0](4,2){\tiny $(0,2)$}
\uput[180](1,3){\tiny $(2,1)$}
\uput[0](3,3){\tiny $(1,2)$}
\uput[90](2,4){\tiny $(2,2)$}
\rput(2,-1.2){$L=3^2$}
\endpspicture
\hspace*{2cm}
\psset{unit=0.5cm}
\pspicture(0,0)(8,9)
\pspolygon(4,0)(6,2)(4,4)(2,2)
\pspolygon(6,2)(8,4)(6,6)(4,4)
\pspolygon(4,4)(6,6)(4,8)(2,6)
\pspolygon[fillstyle=solid,fillcolor=lightgray](2,2)(4,4)(2,6)(0,4)
\psline[linewidth=0.5pt](3,1)(7,5)
\psline[linewidth=0.5pt](1,3)(5,7)
\psline[linewidth=0.5pt](1,5)(5,1)
\psline[linewidth=0.5pt](3,7)(7,3)
\psline[linewidth=2pt](4,4)(3,5)(1,3)(0,4)
\pscircle[fillstyle=solid](4,0){0.1}
\pscircle[fillstyle=solid](3,1){0.1}
\pscircle[fillstyle=solid](5,1){0.1}
\pscircle[fillstyle=solid](2,2){0.1}
\pscircle[fillstyle=solid](4,2){0.1}
\pscircle[fillstyle=solid](6,2){0.1}
\pscircle[fillstyle=solid](1,3){0.1}
\pscircle[fillstyle=solid](3,3){0.1}
\pscircle[fillstyle=solid](5,3){0.1}
\pscircle[fillstyle=solid](7,3){0.1}
\pscircle[fillstyle=solid](0,4){0.1}
\pscircle[fillstyle=solid](2,4){0.1}
\pscircle[fillstyle=solid](4,4){0.1}
\pscircle[fillstyle=solid](6,4){0.1}
\pscircle[fillstyle=solid](8,4){0.1}
\pscircle[fillstyle=solid](1,5){0.1}
\pscircle[fillstyle=solid](3,5){0.1}
\pscircle[fillstyle=solid](5,5){0.1}
\pscircle[fillstyle=solid](7,5){0.1}
\pscircle[fillstyle=solid](2,6){0.1}
\pscircle[fillstyle=solid](4,6){0.1}
\pscircle[fillstyle=solid](6,6){0.1}
\pscircle[fillstyle=solid](3,7){0.1}
\pscircle[fillstyle=solid](5,7){0.1}
\pscircle[fillstyle=solid](4,8){0.1}
\uput[-90](4,0){\tiny $(00,22)$}
\uput[0](6,2){\tiny $(00,20)$}
\uput[0](8,4){\tiny $(02,20)$}
\uput[0](6,6){\tiny $(02,00)$}
\uput[90](4,8){\tiny $(22,00)$}
\uput[180](2,6){\tiny $(20,00)$}
\uput[180](0,4){\tiny $(20,02)$}
\uput[180](2,2){\tiny $(00,02)$}
\uput[0](4,4){\tiny $(00,00)$}
\uput[90](2,4){\tiny $(10,01)$}
\uput[90](3,5){\tiny $(10,00)$}
\uput[180](1,5){\tiny $(20,01)$}
\uput[180](1,3){\tiny $(10,02)$}
\uput[0](3,3){\tiny $(00,01)$}
\rput(4,-1.2){$L=\widetilde{3^2}$}
\endpspicture

\end{center}
\caption{The lattice $L=3^2$ and its bipolar extension. In bold the maximal
  chain corresponding to $f$}
\label{fig:ex3}
\end{figure}

\section{Concluding remarks}
We have provided a general scheme for the bipolarization of a class of posets,
precisely of inf-semilattices. The bipolarization is particularly simple in the
case of a finite distributive lattice $L$, where all connected components of the
poset of join-irreducible elements have a least element (this is the case, e.g., for
Boolean lattices and products of linear lattices). In this case, the
bipolarization of $L$ is made from copies of $L$, and is called for this reason
a regular mosaic. 

Using the concepts of geometric realization of a distributive lattive and of
natural triangulation, we have provided a general interpolation scheme on
bipolar structures, which can be considered as a general definition of the
Choquet integral. 

We have applied our general scheme to multicriteria decision making, where we
have shown that our model reduces to the Choquet integral for $k$-ary
capacities, in the special case where the scores assigned to levels for each
criterion are either 0 or 1, except on one level. We have also provided an
interpretation pertaining to fuzzy set theory when no such restriction exists on
the scores. 

\bibliographystyle{plain}

\bibliography{../BIB/fuzzy,../BIB/grabisch,../BIB/general}

\end{document}